\documentclass[review]{elsarticle}

\usepackage{lineno,hyperref}

\usepackage{algorithm}
\usepackage{algpseudocode}
\usepackage{amsthm}
\usepackage{mathtools}
\usepackage{caption}
\usepackage{subcaption}
\usepackage{pgfplots}
\usepackage{amssymb}
\usepackage{multirow}

\hypersetup{hidelinks}

\makeatletter
\DeclareRobustCommand\onedot{\futurelet\@let@token\@onedot}
\def\@onedot{\ifx\@let@token.\else.\null\fi\xspace}

\makeatother

\usepackage{xcolor}

\pgfplotsset{compat=1.16}

\newcommand{\rfig}[1]{\autoref{fig:#1}}
\newcommand{\ralg}[1]{\autoref{alg:#1}}
\newcommand{\rthm}[1]{\autoref{thm:#1}}

\algnewcommand\Null{\textsc{null }}
\algnewcommand\algorithmicinput{\textbf{Input:}}
\algnewcommand\Input{\item[\algorithmicinput]}
\algnewcommand\algorithmicoutput{\textbf{Output:}}
\algnewcommand\Output{\item[\algorithmicoutput]}
\algnewcommand\algorithmicbreak{\textbf{break}}
\algnewcommand\Break{\algorithmicbreak}
\algnewcommand\algorithmiccontinue{\textbf{continue}}
\algnewcommand\Continue{\algorithmiccontinue}
\algnewcommand{\LeftCom}[1]{\State $\triangleright$ #1}

\newtheorem{thm}{Theorem}[section]
\newtheorem{lem}{Definition}[section]

\colorlet{shadecolor}{black!15}

\theoremstyle{definition}

\def\thmautorefname~#1\null{Theorem~#1~\null}
\def\lemautorefname~#1\null{Define~#1~\null}
\def\algorithmautorefname~#1\null{Algorithm~#1~\null}

\makeatletter
\newenvironment{breakablealgorithm}
  {
   \begin{center}
     \refstepcounter{algorithm}
     \hrule height.8pt depth0pt \kern2pt
     \renewcommand{\caption}[2][\relax]{
       {\raggedright\textbf{\ALG@name~\thealgorithm} ##2\par}%
       \ifx\relax##1\relax 
         \addcontentsline{loa}{algorithm}{\protect\numberline{\thealgorithm}##2}%
       \else 
         \addcontentsline{loa}{algorithm}{\protect\numberline{\thealgorithm}##1}%
       \fi
       \kern2pt\hrule\kern2pt
     }
  }{
     \kern2pt\hrule\relax
   \end{center}
  }
\makeatother

\journal{Journal of \LaTeX\ Templates}

\biboptions{authoryear}

\begin{document}
\sloppy{}

\begin{frontmatter}
 
\title{An Efficient Algorithm for the Partitioning Min-Max Weighted Matching Problem}

\author{Yuxuan Wang, Jinyao Xie, Jiongzhi Zheng, Kun He$^{*}$}
\address{School of Computer Science and Technology, Huazhong University of Science and Technology, Wuhan 430074, China}
\fntext[myfootnote]{The first two authors contribute equally. Corresponding author, Kun He, Email: brooklet60@hust.edu.cn.}



\begin{abstract}

The Partitioning Min-Max Weighted Matching (PMMWM) problem is an NP-hard problem that combines the problem of partitioning a group of vertices of a bipartite graph into disjoint subsets with limited size and the classical Min-Max Weighted Matching (MMWM) problem. Kress et al. proposed this problem in 2015 and they also provided several algorithms, among which MP$_{\text{LS}}$ is the state-of-the-art. In this work, we observe there is a time bottleneck in the matching phase of MP$_{\text{LS}}$. Hence, we optimize the redundant operations during the matching iterations, and propose an efficient algorithm called the MP$_{\text{KM-M}}$ that greatly speeds up MP$_{\text{LS}}$. The bottleneck time complexity is optimized from $O(n^3)$ to $O(n^2)$. We also prove the correctness of MP$_{\text{KM-M}}$ by the primal-dual method. To test the performance on diverse instances, we generate various types and sizes of benchmarks, and carried out an extensive computational study on the performance of MP$_{\text{KM-M}}$ and MP$_{\text{LS}}$. The evaluation results show that our MP$_{\text{KM-M}}$ greatly shortens the runtime as compared with MP$_{\text{LS}}$ while yielding the same solution quality.
	
\end{abstract}


\end{frontmatter}



\section{Introduction}
In this paper, we consider the Partitioning Min-Max Weighted Matching (PMMWM) problem proposed by \cite{a1}. The PMMWM is defined by a weighted bipartite graph $G(U,V,E)$ with disjoint vertex sets $U$, $V$ and edge set $E = \{e_{uv}|u \in U, v \in V\}$, where $U$ needs to be partitioned into $m$ disjoint partitions with no more than $\bar{u}$ vertices each. Given a maximum matching on $G$, the weight of a partition is defined as the sum of the edge weights for the edges matching the vertices in the partition. The goal of the problem is to find a matching and a partition to minimize the largest weight of the partitions. The PMMWM is an NP-hard problem in the strong sense.

\begin{figure}[H]
    \centering
    \begin{subfigure}[b]{0.48\textwidth}
        \centering
		\setlength{\abovecaptionskip}{0.cm}
		\includegraphics[scale=0.2,trim={0 0 0 0}, clip]{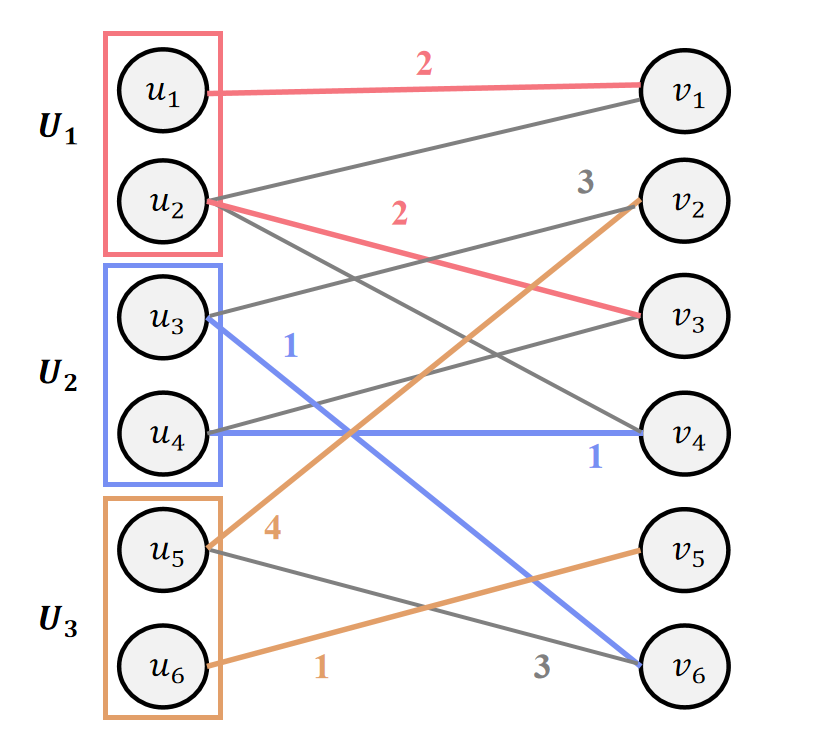}
		\caption{A feasible solution}
		\label{fig:pmmwm1}
    \end{subfigure}
    \hfill
   \begin{subfigure}[b]{0.48\textwidth}
        \centering
		\setlength{\abovecaptionskip}{0.cm}
		\includegraphics[scale=0.2,trim={0 0 0 0}, clip]{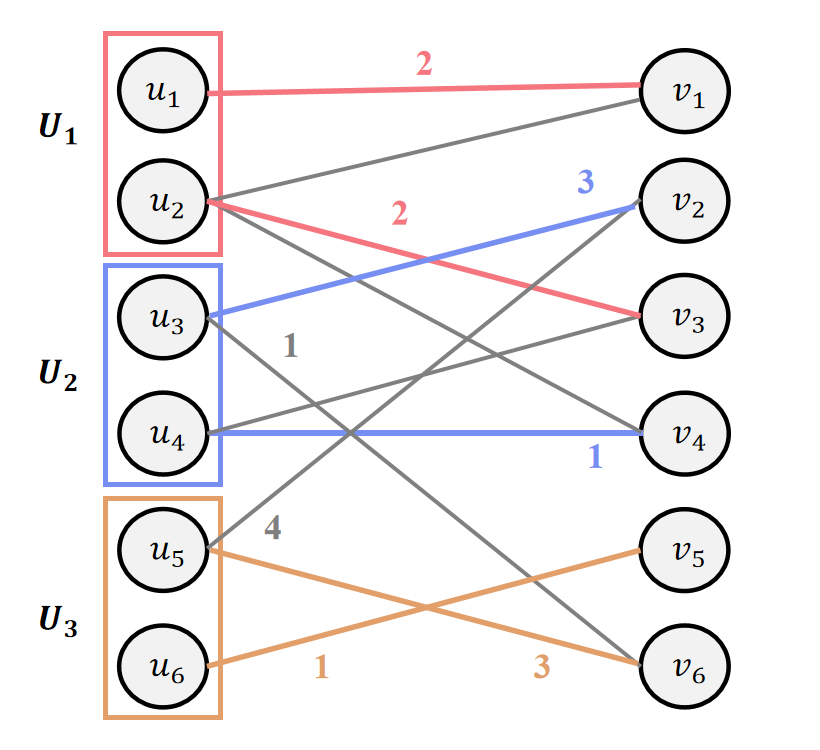}
		\caption{An improved solution}
		\label{fig:pmmwm2}
    \end{subfigure}
\caption{A PMMWM example.}
\label{fig:pmmwm}   
\end{figure}

\rfig{pmmwm1} illustrates a feasible solution to a PMMWM instance. The colored edges indicate the maximum matching edges and the associated weights are marked beside. The left part of the bipartite graph, $U=\{u_1,u_2,...,u_6\}$, is divided into three partitions, $U_1$, $U_2$ and $U_3$ labeled with different colors. The weights of the three partitions are 4, 2 and 5. The value of the objective function is 5. 

We consider two different ways to optimize the objective function.
The first is to modify the partition strategy. We can move $u_6$ to $U_2$, then the weights of the three partitions become 4, 3, and 4. The objective function value is 4. 
The second is to modify the matching strategy. For $u_3,u_5$, if we replace the current matching edges $e_{u_3 v_6}, e_{u_5 v_2}$ with edges $e_{u_3 v_2},e_{u_5 v_6}$, then the weights of the partitions become 4, 4, 4, as shown in \rfig{pmmwm2}. Although the total weight of the matching edges is increased, the final objective function is reduced to 4.

\cite{a1} show an application of the PMMWM at small to medium sized seaports. The seaports include long-term storage areas and temporary storage areas. The vertex set $U$ denotes the empty slots in the long-term storage area, the vertex set $V$ denotes the containers in the temporary storage area, and the edge weight $c_{uv}$ for an edge $e_{uv} \in E$ denotes the cost required to move the container $v$ to the empty slot $u$. The parameter $m$ indicates the number of available telescopic stackers and $\bar{u}$ indicates the upper limit of the number of containers that can be handled by one telescopic stacker. Suppose the cost of a telescopic stacker equals to the total cost of moving its containers, the goal of the problem is to minimize the maximum cost among the telescopic stackers.

We find another application of the PMMWM that is related to the task allocation. Suppose there are $n$ workers, $n$ tasks, and $m$ machines in a workshop. Each task should be performed on a machine by a worker, each worker performs only one task, and each machine can take at most one task at a time. The vertex set $U$ denotes the workers, $V$ denotes the tasks, and the edge weight $c_{uv}$ denotes the spending time of performing task $v$ by worker $u$. The parameter $\bar{u}$ indicates the maximum number of tasks that can be performed by each machine per day (working period). The goal of the problem is to determine the matching between workers and tasks, and the partition of workers to machines, so as to minimize the longest spending time among the machines. 

In this work, we analyze the MP$_{\text{LS}}$ algorithm proposed by \cite{a1}
and find that the time bottleneck is in the matching phase. MP$_{\text{LS}}$ is an iterative algorithm, and at each iteration it adjusts the edge weight matrix of the bipartite graph to generate a different initial maximum matching by an exact algorithm. Since the modified bipartite graph does not change much at each iteration, there are many redundant operations that can be optimized. 
Hence, we propose a new algorithm called MP$_{\text{KM-M}}$, which reuses the matching results of the edges that their weights are not changed in the last iteration, and hence reduces the time complexity from $O(n^3)$ to $O(n^2)$ of the matching phase while maintaining the resulting quality.

The remainder of the paper is organized as follows. We describe the related works in Section 2, and formulate the PMMWM problem in Section 3. Section 4 briefly introduces the MP$_{\text{LS}}$ algorithm, and Section 5 details the proposed MP$_{\text{KM-M}}$ algorithm and proves the correctness. Section 6 reports the results of the experiment, and Section 7  summarizes the main conclusion of this paper.

\section{Related Works}

The PMMWM is an extension problem of the Min-Max Weighted Matching (MMWM) problem~\citep{a2}. MMWM is motivated by the container transshipment in a rail-road terminal. MMWM has the same objective function as PMMWM, but the partitioning of the vertices is given. So, MMWM only needs to decide the matching result. \cite{a2} prove that the MMWM is NP-hard in the strong sense and develop heuristics to find the matching that the average gap to the optimal solution is less than 1\%. 

When there is only one partition (i.e., $m = 1$), the PMMWM can be reduced to an assignment problem that just needs to find a maximum matching with minimum weight. For the assignment problem, \cite{a3} propose an exact polynomial-time algorithm called the Hungarian method which is also denoted as the KM algorithm. The KM algorithm is applied on the matching process of both MP$_{\text{LS}}$~\citep{a1} and MP$_{\text{KM-M}}$. Moreover, the KM algorithm inspired \cite{primal-dual} to propose the primal-dual method for linear programming, which is widely used for various combinatorial optimization problems. In this paper, we also use the primal-dual method to prove the correctness of the proposed MP$_{\text{KM-M}}$ algorithm. 

However, in some practical applications, the partition also needs to be determined, for instance in the small to medium sized seaports where containers are handled by reach stackers. So, \cite{a1} propose the PMMWM and decompose the problem into two stages, matching and partitioning, and solve them either exactly or heuristically. They propose two different heuristic frameworks with different order of the matching and partitioning phases. For the partitioning-matching framework, they propose two heuristics called PM$_{\text{BPS}}$ and PM$_{\text{REG}}$. For the matching-partitioning framework, they propose other two heuristics called MP and MP$_{\text{LS}}$. By performing extensive experiments, \cite{a1} conclude that the MP$_{\text{LS}}$ algorithm is the best-performing one among the four algorithms, and is also robust against changes in the graph structure.

\cite{7376822} also apply PMMWM to the intermodal transport, and propose a heuristic based on tabu search~\citep{DBLP:journals/informs/Glover89,DBLP:journals/informs/Glover90} and a genetic algorithm~\citep{DBLP:books/daglib/0019083}. They propose three neighborhood structures for the tabu search heuristic and a crossover operator for the genetic algorithm. But these two algorithms are not as effective as the MP$_{\text{LS}}$ algorithm. Therefore, MP$_{\text{LS}}$ is still the best-performing algorithm for the PMMWM. 
In this work, we select MP$_{\text{LS}}$ as the baseline algorithm, and our proposed MP$_{\text{KM-M}}$ algorithm significantly outperforms MP$_{\text{LS}}$ on the algorithm efficiency while keeping the same effectiveness.

\section{Problem Formulation}
Let $G(U,V,E)$ be a weighted bipartite graph, where $U$ and $V$ are two disjoint vertex sets, and $E = \{e_{uv}|u \in U, v \in V\}$ denotes the set of pairwise edges. 
Assume $|U| = n_1$, $|V| = n_2$, and $n_1 \leq n_2$. Each edge $e_{uv} \in E$ is associated with a weight $c_{uv}\in\mathbb{Q}^{+}_{0}$. Define a matching as the set $M\subseteq E$ of pairwise nonadjacent edges, and the maximum matching~\citep{DBLP:journals/computing/DerigsZ78} as the matching with the largest $|M|$ among all matchings on $G$. 


Assume that for any given weighted bipartite graph, there exists at least one maximum matching $\Pi$ such that $|\Pi|=n_{1}$. Define a partition $\mathcal{P}$ of $U$ that divides $U$ into $m$ disjoint partitions $U_1,U_2,...,U_m$, with at most $\bar{u}$ vertices in each partition. The goal of the PMMWM is to find a partition $\mathcal{P}$ of $U$ and a maximum matching $\Pi$ of $G$ that minimizes the evaluation function $f(\mathcal{P}, \Pi) := \max_{k\in\{1, ...,m\}}\{\sum_{u\in U_k,(u,v)\in\Pi}c_{uv}\}$.



Let $z_{uvk}$ be a binary variable such that $z_{uvk} = 1$ if $u \in U_k, e_{uv} \in \Pi$, and $z_{uvk} = 0$ otherwise. 
The PMMWM problem can be formalized as follows. 

\begin{equation*}
\begin{aligned}
    \label{eq:pmmwm}
    & \min_{\mathbf{z}} \max_{k\in\{1,…,m\}} \{\sum_{u\in U}\sum_{v\in V}c_{uv}z_{uvk}\} \\
\text{s.t.} \qquad & (1) \qquad\sum_{k=1}^{m}\sum_{v\in V}z_{uvk}=1\quad\forall u\in U,\\
    &(2) \qquad \sum_{k=1}^{m}\sum_{u\in U}z_{uvk}\le1\quad\forall v \in V,\\
    &(3) \qquad \sum_{k=1}^{m}z_{uvk}=1\quad\forall e_{uv} \in E,\\
    &(4) \qquad \sum_{u \in U}\sum_{v\in V}z_{uvk}\le\bar{u}\quad\forall k\in\{1,…,m\}.
\end{aligned}
\end{equation*}

The objective function minimizes the maximum matching weight among all possible assignments and maximum matchings. Constraints (1) and (2) are the well-known constraints on the maximum matching. Constraint (3) forces each vertex $u\in U$ to belong to exactly one partition $U_k,k\in \{1,...,m\}$. Constraint (4) restricts that the size of the partition does not exceed $\bar{u}$. 

\section{The MP$_{\text{LS}}$ Algorithm}

The MP$_{\text{LS}}$ algorithm~\citep{a1} decomposes the PMMWM problem into two components, i.e., matching and partition. MP$_{\text{LS}}$ is an iterative heuristic algorithm, and each iteration consists of three stages. The first stage generates a maximum matching $\Pi$ of $G = (U,V,E)$, and the second stage calculates a partition $\mathcal{P}$ of $U$ that satisfies constraints (3) and (4) of the problem. Then the combination of $\mathcal{P}$ and $\Pi$ results in a feasible solution for an input instance. In the third stage, MP$_{\text{LS}}$ makes a slight adjustment on the weight matrix of $E$ and then returns to the first stage of the next iteration, which can help the algorithm to escape from the local optima and find better solutions. We denote $G'$ to be the bipartite graph where the edge weight matrix is adjusted. Starting from $G' = G$, the MP$_{\text{LS}}$ algorithm repeats the three stages until a termination condition is met. Implementation details of the three stages are as follows.

In the first stage, MP$_{\text{LS}}$ applies the KM algorithm~\citep{DBLP:books/daglib/p/Kuhn10} to obtain a minimum weight maximum matching $\Pi$ of the bipartite graph $G'$. 
So match $\Pi$ is a maximum matching such that the total weight of the edges belonging to the matching is minimized.

In the second stage, MP$_{\text{LS}}$ first assumes the maximum matching $\Pi$ of $G'$ is fixed, then the problem can be regarded as a restricted partitioning problem of $U$, by setting the weight of each vertex $u \in U$ to $w_u = c_{uv}, e_{uv} \in \Pi$ (note that when calculating the partition in this stage, $c_{uv}$ is always the original weight of edge $e_{uv}$). MP$_{\text{LS}}$ uses the RPH~\citep{a2} heuristic algorithm to solve the restricted partitioning problem and obtains a feasible solution to the PMMWM. 

In the third stage, MP$_{\text{LS}}$ changes the weight of one edge in $G'$ so as to adjust the bipartite matching solution in the next iteration. Let $k = \mathop{\arg\max_{k \in \{1,...,m\}}{\{\sum_{u \in U_k, e_{uv} \in \Pi}{c_{uv}}\}}}$, then the selected edge $e \in \Pi$ is an edge with the maximum weight among the current edges of $\Pi$ that are incident to vertices of component $U_k$. MP$_{\text{LS}}$ changes the weight of $e$ in $G'$ to a large value ($10^2 c_{\max}$ in default, where $c_{\max}$ is the maximum edge weight in the input graph) to largely reduce the probability that this edge will belong to the matching calculated in the first stage of the next iteration. 

The MP$_{\text{LS}}$ algorithm stops when the solution to the PMMWM has not improved for 20 iterations.





\section{The Proposed MP$_{\text{KM-M}}$ Algorithm}
In this section, we propose an MP$_{\text{KM-M}}$ algorithm that significantly improves the efficiency of the SOTA algorithm, MP$_{\text{LS}}$, for solving the PMMWM. 
We observe that in each iteration, MP$_{\text{LS}}$ only changes the weight of one edge in the third stage, but recalculates the minimum weight maximum matching of the whole graph in the first stage of the next iteration, which is redundant and very time-consuming. 
Specifically, the KM algorithm used in the first stage of MP$_{\text{LS}}$ will cost $O(n^3)$ time in each iteration. To handle this issue, the proposed MP$_{\text{KM-M}}$ optimizes the efficiency of the first stage by simplifying the matching process.

We first introduce the KM algorithm in detail, then presents the simplified matching algorithm KM-M in the proposed MP$_{\text{KM-M}}$ algorithm and the main process of MP$_{\text{KM-M}}$, and finally provides the proof on the correctness of the KM-M algorithm.


\subsection{The KM Algorithm}

The KM algorithm~\citep{DBLP:books/daglib/p/Kuhn10} is mainly used to solve the maximum weight maximum bipartite matching problem. By multiplying the weight of each edge in the input bipartite graph by $-1$, KM can be used to solve the minimum weight maximum bipartite matching problem in the first stage of MP$_{\text{LS}}$. Before introducing the procedure of the KM algorithm, we first introduce some definitions.

\begin{lem}
	Alternating Path. Given a bipartite graph $G = (U,V,E)$ with some edges matched. An alternating path on $G$ is a path that starts from an unmatched vertex, then traverses an unmatched edge and matched edge alternatively.
\end{lem}
	
\begin{lem}
	Augmenting Path. Given a bipartite graph $G = (U,V,E)$ with some edges matched. An augmenting path on $G$ is a special alternating path on $G$ whose starting and ending vertices are both unmatched.
\end{lem}
	
\begin{lem}
	Feasible Label. A label of a bipartite graph $G = (U,V,E)$ can be represented by assigning a value $ex_i$ on each vertex $i \in U \cup V$. A feasible label is a label that satisfies $ex_u+ex_v \ge c_{uv}$ for each edge $e_{uv} \in E$.
\end{lem}
	
\begin{lem}
	Equivalence Subgraph. An equivalence subgraph is a spanning subgraph of the original graph (a spanning subgraph contains all nodes of the original graph, but not all edges) that only contains edges satisfying $ex_u+ex_v=c_{uv}$.
\end{lem}

The procedure of the KM algorithm is shown in Algorithm \ref{alg:km}. The KM algorithm first initializes the value of each vertex. The KM algorithm sets $ex_{u} = \max_{v\in V, e_{uv}\in E}\{c_{uv}\}$ for each vertex $u \in U$, and sets $ex_v = 0$ for each vertex $v \in V$. Then KM traverses all the vertices in $U$ and tries to match each vertex $u$ by finding an augmenting path in the equivalence subgraph starting from $u$.


\begin{breakablealgorithm}
	\caption{The KM algorithm}
	\label{alg:km}	
	\begin{algorithmic}[1]
		\Input $G(U,V,E)$
		\Output $\Pi$
		\Function{KM}{$G(U,V,E)$}
		\State $ex_{v \in V} \gets 0,ex_{u\in U} \gets \max_{v\in V, e_{uv}\in E}\{c_{uv}\}$, $\Pi \gets \emptyset$
		\For{$u \in U$}
		\State $\Pi \gets$ MATCH($u$, $\Pi$)
		\EndFor
		\State \Return $\Pi$
		\EndFunction
	\end{algorithmic}
\end{breakablealgorithm} 

The algorithm for matching a vertex $u \in U$ is shown in Algorithm \ref{alg:match}. The matching algorithm uses the \textit{findpath}() function (Algorithm \ref{alg:findpath}) to find an augmenting path in the equivalence subgraph from the input vertex $u$ (line 5). Once an augmenting path is found, each vertex belonging to $U$ in the augmenting path will be matched with a vertex belonging to $V$ in the augmenting path. Specifically, for an augmenting path $\{u_1,v_1,u_2,v_2,...,u_k,v_k\}$, the vertex $u_i$ is matched with the vertex $v_i$ ($i \in \{1,...,k\}$).

If the \textit{findpath}() function cannot find an augmenting path in the equivalence subgraph, then the matching algorithm tries to adjust the equivalence subgraph by adjusting the feasible label of the input bipartite graph (lines 8-18). Let $vis_x = 1$ indicates vertex $x$ is on the alternating path generated by the last search, and $vis_x = 0$ otherwise. Let $\Delta=\min_{u \in U, v \in V}\{vis_u vis_v(ex_u+ex_v-c_{uv})\}$ denotes the smallest adjustment value that makes it possible for the algorithm to find an augmenting path in the equivalence subgraph. Then the label is adjusted as follows: for each vertex $u \in U \wedge vis_u = 1$, $ex_u := ex_u-\Delta$; for each vertex $v \in V \wedge vis_v = 1$, $ex_v := ex_v+\Delta$.

Moreover, the KM algorithm designs a relaxation function $slack$ to optimize the efficiency of calculating $\Delta$ (line 8 in Algorithm \ref{alg:match}). For each vertex $v \in V$, $slack_v = \min_{u \in U, e_{uv} \in E}\{ex_u+ex_v-c_{uv}\}$ denotes the smallest value that increasing $ex_v$ by this value enables at least one edge connected with $v$ to be appeared in the equivalence subgraph. The relaxation function $slack$ is adjusted adaptively in the matching algorithm and the \textit{findpath} function. By applying the $slack$ function, the time complexity of calculating $\Delta$ can be reduced from $O(n^2)$ to $O(n)$, and the time complexity of the entire KM algorithm can be reduced from $O(n^4)$ to $O(n^3)$.




\begin{breakablealgorithm}
	\caption{The matching algorithm}
	\label{alg:match}
	\begin{algorithmic}[1]
		\Input $u\in U$, $\Pi$
		\Output $\Pi$
		\Function{match}{$u$, $\Pi$}
		\State $slack_{v\in V} \gets\infty$
		\While{$u$ is not matched}
		\State $vis_{u\in U,v\in V}\gets\text{false}$
		\If{FINDPATH($u$)}
		\State break
		\EndIf
		\State $\Delta\gets\min_{v\in V,vis_v=false}\{slack_v\}$
		\For{$u\in U, vis_u=true$}
		\State $ex_u\gets ex_u-\Delta$
		\EndFor
		\For{$v\in V$}
		\If{$vis_v = true$}
		\State $ex_v \gets ex_v + \Delta$
		\Else
		\State $slack_v \gets slack_v - \Delta$
		\EndIf
		\EndFor
		\EndWhile
		\State \Return $\Pi$
		\EndFunction
	\end{algorithmic}
\end{breakablealgorithm}



\begin{breakablealgorithm}
	\caption{Finding an augmenting path}
	\label{alg:findpath}
	\begin{algorithmic}[1]
		\Input $u \in U$
		\Output Whether an augmenting path is found or not
		\Function{findpath}{$u$}
		\State $vis_u\gets\text{true}$ 
		\For {$v \in V\text{ and }vis_v=\text{false}$} 
		\If{$ex_u+ex_v=c_{uv}$}
		\State $vis_v\gets\text{true}$, $u' \gets u \in U \cap e_{uv} \in \Pi$ 
		\If{$v \text{ is not matched} \vee  \text{findpath}(u')$}
		\State $\Pi \gets \Pi \backslash \{e_{u'v}\}$, $\Pi \gets \Pi \cup \{e_{uv}\}$
		\State \Return true
		\EndIf
		\Else
		\State $slack_v\gets\min\{slack_v,ex_u+ex_v-c_{uv}\}$
		\EndIf
		\EndFor
		\State \Return false
		\EndFunction
	\end{algorithmic}	
\end{breakablealgorithm}


\subsection{The Simplified KM-M Matching Algorithm}

The MP$_{\text{LS}}$ algorithm performs the KM algorithm on the entire input bipartite graph in the first stage of each iteration, which is redundant. The proposed MP$_{\text{KM-M}}$ algorithm presents a simplified matching algorithm, denoted as KM-M, to calculate the minimum weight maximum matching in the first stage.
The procedure of KM-M is shown in Algorithm \ref{alg:km-m}. KM-M receives a bipartite graph $G$, a maximum matching $\Pi$ and a vertex $u \in U$ as the input, and outputs a maximum matching $\Pi$. We then explain how the KM-M algorithm works in the proposed MP$_{\text{KM-M}}$ algorithm by introducing the process of MP$_{\text{KM-M}}$ as follows.

MP$_{\text{KM-M}}$ is also an iterative heuristic algorithm and each iteration also consists of three stages like MP$_{\text{LS}}$ does. In the first stage of the first iteration of MP$_{\text{KM-M}}$, MP$_{\text{KM-M}}$ sends the original bipartite graph $G$, an empty maximum matching $\Pi = \emptyset$, and an arbitrary vertex $u \in U$ to the KM-M algorithm. Then KM-M performs the KM algorithm on $G$ and outputs a maximum matching $\Pi$ as MP$_{\text{LS}}$ does. 

The second and third stages of MP$_{\text{KM-M}}$ perform the same operations as  MP$_{\text{LS}}$ does, i.e., generates a solution to the PMMWM and improves the solution by a local search in the second stage, and adjusts the weight of an edge to $10^2c_{max}$ in the third stage. Suppose $\Pi$ is the maximum matching calculated by the first and second stages (local search in the second stage may adjust the maximum matching generated by KM-M), $e_{uv}$ ($u \in U, v \in V$) is the edge selected in the third stage, and $G'$ is the bipartite graph where the weight of edge $e_{uv}$ is just adjusted. MP$_{\text{KM-M}}$ will send $G'$, $\Pi$ and $u$ to the KM-M algorithm in the next iteration, and the KM-M algorithm only needs to remove the matching of $u$ in $\Pi$, and rematch $u$ in $G'$, which is significantly more efficient than performing KM on the entire $G'$. 

As a result, the proposed MP$_{\text{KM-M}}$ can reduce the time complexity of the first stage in each iteration except the first iteration of MP$_{\text{LS}}$ from $O(n^3)$ to $O(n^2)$. Moreover, the termination condition of MP$_{\text{KM-M}}$ is the same as that of MP$_{\text{LS}}$, i.e., MP$_{\text{KM-M}}$ stops when the solution to the PMMWM has not improved for 20 iterations.


\begin{breakablealgorithm} 
	\caption{The KM-M algorithm}
	\label{alg:km-m}	
	\begin{algorithmic}[1]
		\Input $G(U,V,E)$, $\Pi$, $u$
		\Output $\Pi$
		\Function{KM-M}{$G(U,V,E)$, $\Pi$, $u$}
		\If{$\Pi \neq \emptyset$}
		\State $e_{uv} \gets$ the edge in $\Pi$ and is connected with $u$ 
		\State $\Pi \gets \Pi \backslash \{e_{uv}\}$, $\Pi \gets$ match($u$, $\Pi$)
		\Else
		\State $ex_{v \in V} \gets 0,ex_{u\in U} \gets \max_{v\in V, e_{uv}\in E}\{c_{uv}\}$, $\Pi \gets \emptyset$
		\EndIf
		\For{$u \in U$}
		\State $\Pi \gets$ match($u$, $\Pi$)
		\EndFor
		\State \Return $\Pi$
		\EndFunction
	\end{algorithmic}
\end{breakablealgorithm}

\rfig{km-m1} illustrates a maximum weight maximum matching $\Pi$ on the graph $G$ calculated by the KM algorithm. The red and black edges satisfy $ex_u+ex_v=c_{uv}$ (i.e., edges in the equivalence subgraph), and the red edges is belonging to $\Pi$. Suppose we modify the weight of the matching edge $e_{u_1 v_1}$ from $-2$ to $-10^2c_{max}$ ($-500$) and yield the new graph $G'$. The KM-M algorithm performs as follows. Firstly, as \rfig{km-m2} shows, KM-M removes the matching of $u_1$ from $\Pi$ and the modified edge $e_{u_1 v_1}$ from the equivalence subgraph. Then KM-M uses the matching algorithm (Algorithm \ref{alg:match}) to find a matching edge for $u_1$. As \rfig{km-m3} shows, the algorithm adjusts the feasible label to add edge $e_{u_1 v_2}$ into the equivalence subgraph. Finally, as shown in \rfig{km-m4}, we find an augmenting path $\{u_1,v_2,u_3,v_1\}$ (the dashed edges) and replace the edge $e_{u_3 v_2}$ with $e_{u_1 v_2},e_{u_3 v_1}$ in the matching $\Pi$. The new matching is as good as the output of the KM algorithm on the entire graph $G'$. Obviously, KM-M reduces a lot of redundant operations for calculating the matching of other edges. 

\begin{figure}[H]
    \centering
    \begin{subfigure}[b]{0.48\textwidth}
        \centering
		\setlength{\abovecaptionskip}{0.cm}
		\includegraphics[scale=0.2,trim={0 0 0 0}, clip]{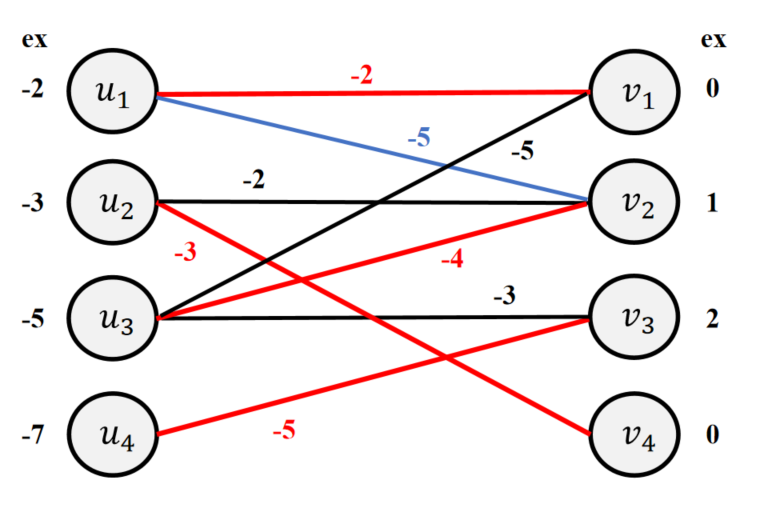}
		\caption{A matching $\Pi$ on $G$}
		\label{fig:km-m1}
    \end{subfigure}
    \hfill
   \begin{subfigure}[b]{0.48\textwidth}
        \centering
		\setlength{\abovecaptionskip}{0.cm}
		\includegraphics[scale=0.2,trim={0 0 0 0}, clip]{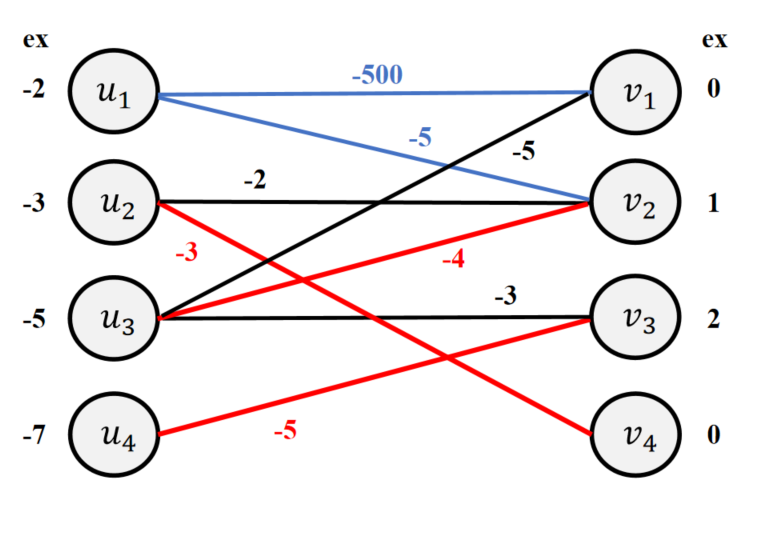}
		\caption{Modify the weight of an edge}
		\label{fig:km-m2}
    \end{subfigure}
    \begin{subfigure}[b]{0.48\textwidth}
        \centering
		\setlength{\abovecaptionskip}{0.cm}
		\includegraphics[scale=0.2,trim={0 0 0 0}, clip]{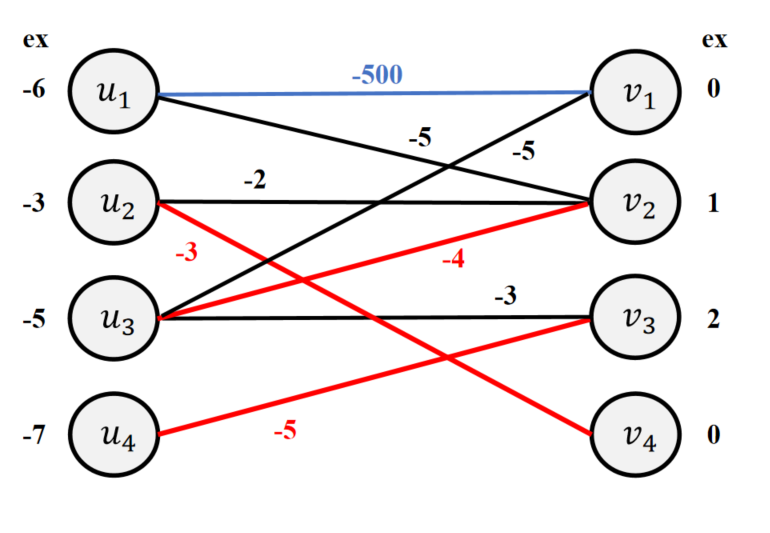}
		\caption{Adjust the feasible label}
		\label{fig:km-m3}
    \end{subfigure}
    \hfill
   \begin{subfigure}[b]{0.48\textwidth}
        \centering
		\setlength{\abovecaptionskip}{0.cm}
		\includegraphics[scale=0.2,trim={0 0 0 0}, clip]{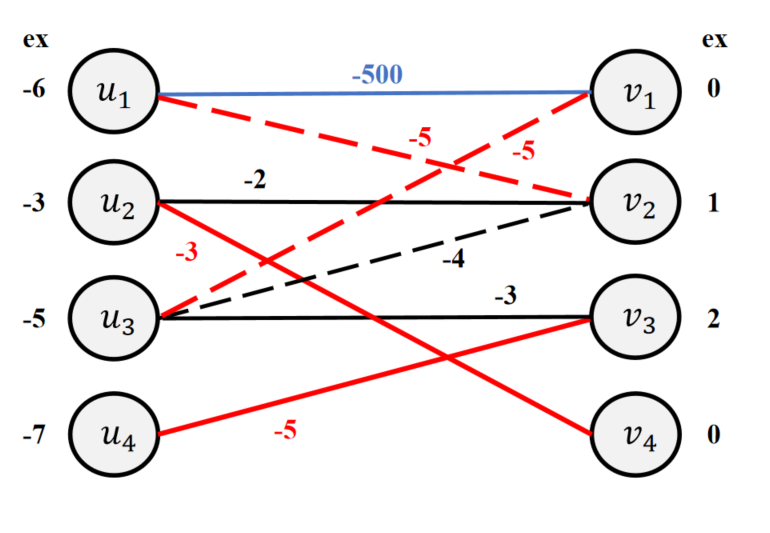}
		\caption{Find an augmenting path}
		\label{fig:km-m4}
    \end{subfigure}
\caption{An illustrate of the KM-M procedure.}
\label{fig:km-m}   
\end{figure}


\subsubsection{The Correctness of KM-M Algorithm}

To show the correctness of KM-M, we need to briefly review the KM algorithm by the primal-dual method~\citep{DBLP:books/daglib/0069809}. The dual of the linear program of the maximum weight maximum matching problem is to minimize the sum of the feasible label.

We first present some conceptions. The primal solution is a matching $\Pi$. The dual solution is an assignment of the feasible label satisfying $ex_u+ex_v \ge c_{uv}$ for each $e_{uv} \in E$. If the primal and dual solutions obey the complementary slackness, they are optimal solutions. The complementary slackness means that each edge $e_{uv} \in \Pi$ satisfies $ex_u+ex_v = c_{uv}$, i.e., $\Pi$ is a perfect matching in the equivalence subgraph of the dual solution.

KM optimizes the dual solution until finding a perfect match in the equivalence subgraph. If a perfect matching is found, this matching is an optimal primal solution (maximum weight maximum matching) and the algorithm ends. Otherwise, there is a modification method~\citep{DBLP:journals/jacm/EdmondsK72} to decrease the sum of the feasible label and keep the dual solution feasible, which also adds new edges to the equivalence subgraph.


Now let's get back to the correctness of KM-M. We first present the theorem, and then the proof.

\begin{thm}
\label{thm:km-m}
Given a bipartite graph $G = (U,V,E)$ and a maximum weight maximum matching $\Pi$ calculated by the KM algorithm. If we modify the weight of exactly one edge $e_{uv} \in \Pi$ from $c_{uv}$ to a relatively large value, say $10^2c_{max}$, we only need to remove $e_{uv}$ from $\Pi$, and rematch $u$ in the new graph, to obtain the maximum weight maximum matching in the new graph.
\end{thm}

\begin{proof}

We first prove the feasibility of the dual solution. Because we only modify the weight of edge $e_{uv} \in \Pi$ from $c_{uv}$ to $10^2c_{max}$, and we known $ex_{u}+ex_{v}= -c_{uv}$, thus $ex_{u}+ex_{v} > -c^{\prime}_{u,v} = -10^2c_{max}$ holds. So the dual solution is still feasible for the current matching, indicating that the previous information can be kept. After modifying the weight of $e_{uv} \in \Pi$, only vertex $u \in U$ does not have a matching. Therefore, instead of executing the entire KM algorithm, we only need to execute \ralg{match} on $u$. 


Then, since \ralg{match} only decreases the sum of feasible labels, we need to prove that after modifying the edge weight of $e_{uv}$, the sum of the matching weight doesn't become larger. This is easy to proof by contradiction. Assuming that the sum increases, there are two possible cases. The first one is that the modified edge is in the matching. As the edge weight decreases, the matching weight of the remaining edges increases, indicating that the previous matching is not the maximum weight maximum matching, which is a contradiction. The second one is that the modified edge is not in the matching, which can be proved in the same way. 
So \rthm{km-m} holds, and therefore KM-M is correct.

\end{proof}

\section{Experimental Results and Comparisons}

To analyze the performance of the proposed MP$_{\text{KM-M}}$ algorithm, we do experiments on extensive instances. All the experiments were performed on an Intel® Core™ i7 CPU at 2.7 GHz and 8GB of RAM, running Windows10 64-bit. All the tested algorithms are implemented in C++.

\subsection{Data Generation Strategy}

 We generate four sets of PMMWM instances by the strategy proposed by \cite{a1}, so as to do a comprehensive comparison with the baseline algorithm, MP$_{\text{LS}}$,. Since $n_1\neq n_2$ can be reduced to the case of $n_1 = n_2$ by adding edges with zero weight. In this paper, we define $n: = n_1 = n_2$, and the four instance sets are defined as follows.

The first instance set contains two classes, BPS70 and BPS80. We first randomly generate $n^2$ rational numbers on the interval [1, 1000] and a complete bipartite graph $G(U,V,E)$ with vertex sets $U=\{u_1,...,u_n\}$ and $V=\{v_1,...,v_n\}$. We sort the rational numbers in non-decreasing order 
and store in 
list $L$. Then we repeat the following operations on each vertex $v_i \in V$ from $v_1$ to $v_n$, to assign $n$ rational numbers in $L$ to the $n$ edges incident with $v_i$ in $E$. 
For $v_i \in V$, we assign the first $\lfloor0.8n\rfloor$ (in case of BPS80) or $\lfloor0.7n\rfloor$ (in case of BPS70) elements of $L$ to the edges $(u_1 ,v_i)$, $(u_2,v_i)$, $...$, $(u_{\lfloor0.8n\rfloor},v_i)$ (BPS80) or $(u_1,v_i)$, $(u_2,v_i)$, $...$, $(u_{\lfloor0.7 n\rfloor},v_i)$ (BPS70) and remove these elements from $L$. 
For the remaining edges incident with $v_i$, we randomly assign the remaining elements in $L$ in turn, and remove the assigned elements from $L$.

The second instance set is RAND. We first generate a complete bipartite graph $G(U,V,E)$ with vertex sets $U=\{u_1,...,u_n\}$ and $V=\{v_1,...,v_n\}$, then enumerate each edge $e_{uv} \in E$ in turn and assign a random integer in [1, 1000] as its weight.

The third instance set contains two classes, SPARSE70 and SPARSE80. Instances belonging to this set is defined on non-complete bipartite graphs with $|E|=\lceil0.7n^2\rceil$ for SPARSE70 or $|E|=\lceil0.8n^2\rceil$ for SPARSE80. To generate an instance of this set, we first randomly generate $\lceil0.7n^2\rceil$ (in case of SPARSE70) or $\lceil0.8n^2\rceil$ (in case of SPARSE80) rational number on the interval [1, 1000], and store in list $L$. Then $e_{u_i v_i},\{i=1,...,n\}$ is added to $E$ to ensure that there must be a feasible solution for the instance. Each edge $e_{u_i v_i}$ is assigned a random element in $L$. Once an element in $L$ is assigned to an edge, the element will be removed from $L$. For each remaining element in $L$, we randomly add an edge $e_{uv} \not\in E$ into $E$ and assign the element as the weight of $e_{uv}$. In the end, we permute the order of $v_i \in V$ to obtain randomness.

Since the density of each instance in the third set is still high, the same construction method of generating the instances in the third set is used to generate the fourth instance set, which contains two classes, SPARSE30 and SPARSE20. Each instance in the fourth set is defined on a sparse bipartite graph with $|E|=\lceil0.3n^2\rceil$ for SPARSE30 or $|E|=\lceil0.2n^2\rceil$ for SPARSE20. 

For each instance set, we generate the instances with different parameter settings. For generating small instances, we vary the size of the bipartite graphs from 100 to 190, i.e., $n=100,110,120, ...,190$. For generating large instances, we vary the size of the bipartite graphs from 1000 to 1400, i.e., $n = 1000,1100,1200,...,1400$. For the number of partitions, $m$, we set it to 2, $\lfloor0.04n\rfloor$, $\lfloor0.08n\rfloor$ or $\lfloor0.125n\rfloor$. We set the maximum number of vertices in each partition, $\bar{u}$ to $\lceil\dfrac{n}{m}\rceil$, $\lfloor\lceil\dfrac{n}{m}\rceil+\dfrac{1}{3}(n-\lceil\dfrac{n}{m}\rceil)\rfloor$ or $n$. 

For each combination of parameters and
each generation strategy, we construct 20 instances if $n < 1000$, otherwise we construct 5 instances. In the end, we have a total of 18,900 instances, containing 16,800 small instances and 2,100 large instances.

\subsection{Comparison of MP$_{\text{KM-M}}$ and MP$_{\text{LS}}$}

We test the proposed MP$_{\text{KM-M}}$ algorithm and MP$_{\text{LS}}$ on all the 18,900 instances. Each instance is solved once by each algorithm. Since MP$_{\text{KM-M}}$ only optimizes the efficiency of the matching process, and will not change the matching result of the exact KM algorithm, the results of MP$_{\text{KM-M}}$ and MP$_{\text{LS}}$ are the same for each of the 18,900 instances. The consistent running results are as expected, which demonstrate the correctness of the proposed KM-M algorithm. This subsection presents the comparison results on the runtime of MP$_{\text{KM-M}}$ and MP$_{\text{LS}}$ on small and large instances respectively.


\subsubsection{Comparison on Small Instances}

We first compare the runtime of MP$_{\text{KM-M}}$ and MP$_{\text{LS}}$ on all the 16,800 small instances. \rfig{st} shows the runtime of each algorithm per instance varying with the problem scale $n$ in the four instance sets respectively. For each $n$, the results are expressed by the average runtime over all parameter settings and all generation strategies. Moreover, \rfig{st5} shows the ratio of the average runtime of MP$_{\text{LS}}$ to MP$_{\text{KM-M}}$ varying with the problem scale $n$.

From the results of \rfig{st} and \rfig{st5} we can see that:

(1) For each problem scale $n$, the proposed MP$_{\text{KM-M}}$ algorithm significantly outperforms MP$_{\text{LS}}$ on the efficiency. Specifically, the runtime of MP$_{\text{LS}}$ is 2 to 12 times longer than that of MP$_{\text{KM-M}}$.

(2) With the increment of $n$, the average runtime of MP$_{\text{LS}}$ increases much quicker than that of MP$_{\text{KM-M}}$. This is because the matching process takes up a large proportion in the entire algorithm for solving small instances, and the runtime of the KM algorithm increases sharply with the increment of $n$. Thus, the improvement of MP$_{\text{KM-M}}$ over MP$_{\text{LS}}$ is more significant for larger $n$.

(3) On the first set of small instances, the runtime of MP$_{\text{LS}}$ is at least 6 times more than that of MP$_{\text{KM-M}}$. The improvement is more significant than that on the other three sets. This is because the runtime of KM is longer for calculating the maximum matching of complete and constructed bipartite graphs.

\begin{figure}[H]
    \centering
    \begin{subfigure}[b]{0.48\textwidth}
        \centering
		\setlength{\abovecaptionskip}{0.cm}
		\includegraphics[scale=0.4,trim={12 5 20 35}, clip]{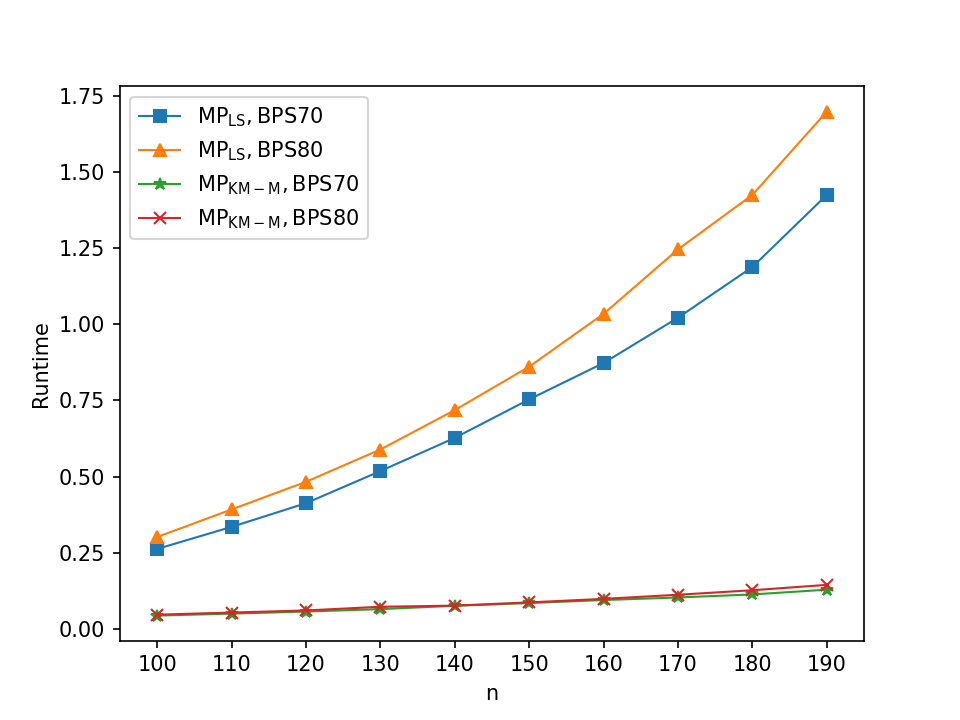}
		\caption{Comparison on the first set}
		\label{fig:st1}
    \end{subfigure}
    \hfill
   \begin{subfigure}[b]{0.48\textwidth}
        \centering
		\setlength{\abovecaptionskip}{0.cm}
		\includegraphics[scale=0.4,trim={12 5 20 35}, clip]{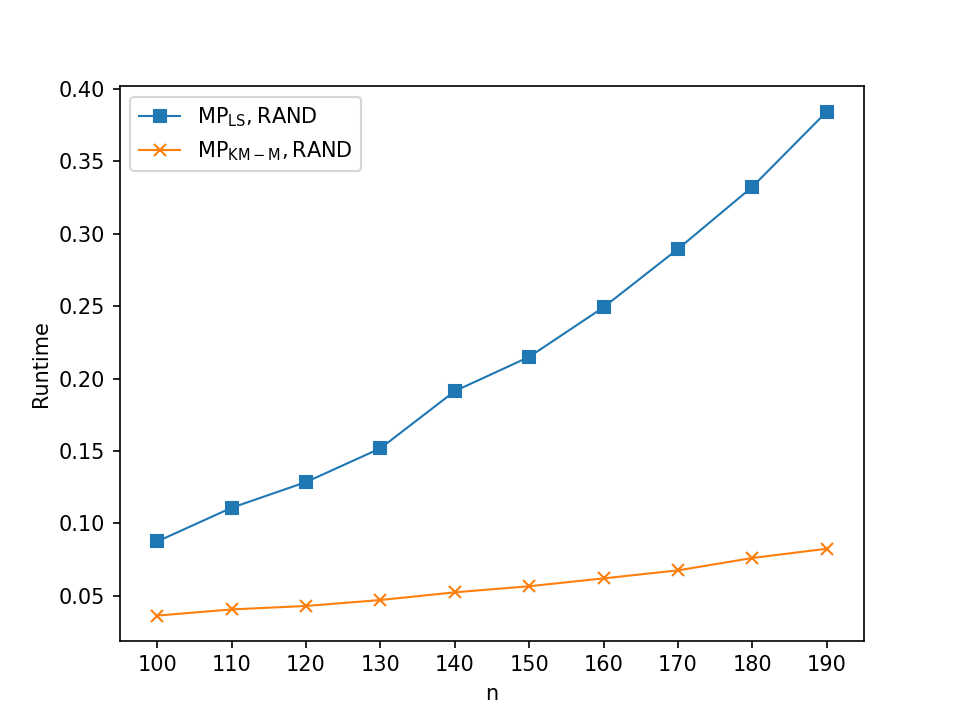}
		\caption{Comparison on the second set}
		\label{fig:st2}
    \end{subfigure}
    \begin{subfigure}[b]{0.48\textwidth}
        \centering
		\setlength{\abovecaptionskip}{0.cm}
		\includegraphics[scale=0.4,trim={12 5 20 35}, clip]{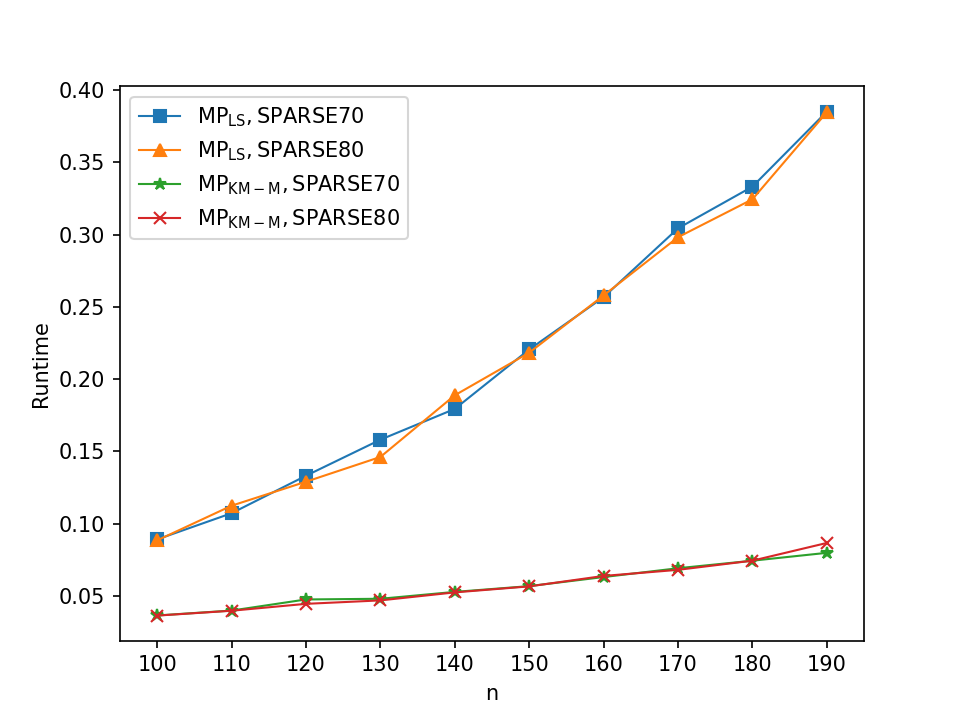}
		\caption{Comparison on the third set}
		\label{fig:st3}
    \end{subfigure}
    \hfill
   \begin{subfigure}[b]{0.48\textwidth}
        \centering
		\setlength{\abovecaptionskip}{0.cm}
		\includegraphics[scale=0.4,trim={12 5 20 35}, clip]{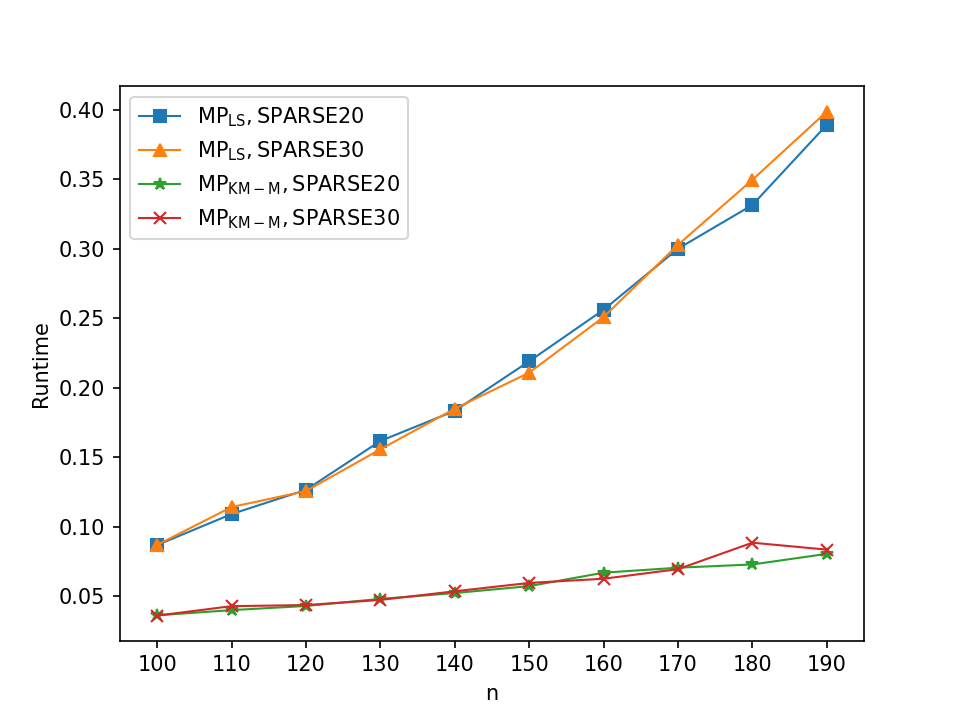}
		\caption{Comparison on the fourth set}
		\label{fig:st4}
    \end{subfigure}
\caption{Comparison on the runtime of MP$_{\text{KM-M}}$ and MP$_{\text{LS}}$ on the small instances.}
\label{fig:st}   
\end{figure}


\begin{figure}[ht]
	\centering
	\setlength{\abovecaptionskip}{0.cm}
	\includegraphics[scale=0.5,trim={20 10 20 40}, clip]{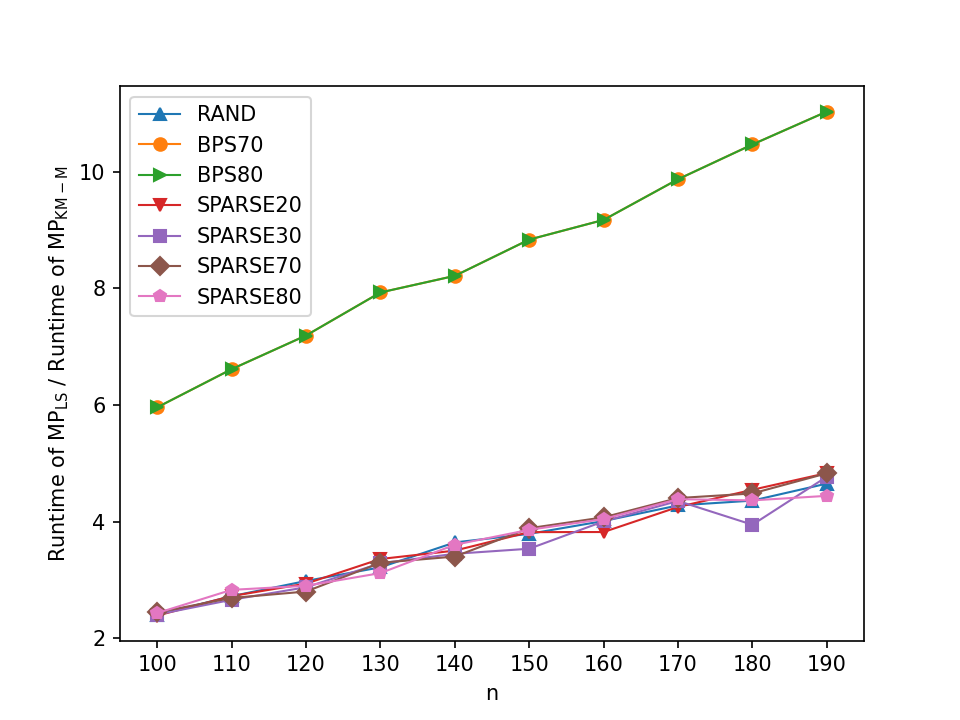}
	\caption{Ratio of the average runtime of MP$_{\text{LS}}$ over MP$_{\text{KM-M}}$ on small instances.}
	\label{fig:st5}
\end{figure}

\subsubsection{Comparison on Large Instances}

We then do comparison on all the 2,100 large instances. \rfig{lt} show the results of average runtime of each algorithm varying with the problem scale $n$ in the four instance sets, respectively. And \rfig{lt} shows the ratio of the average runtime of MP$_{\text{LS}}$ to MP$_{\text{KM-M}}$ varying with the problem scale $n$.

From the results of \rfig{lt} and \rfig{lt5} we can see that:

(1) MP$_{\text{KM-M}}$ also significantly outperforms MP$_{\text{LS}}$ on the efficiency when solving large instances. Specifically, the runtime of MP$_{\text{LS}}$ is 8 to 19 times longer than that of MP$_{\text{KM-M}}$.

(2) The ratio of the runtime of MP$_{\text{LS}}$ over MP$_{\text{KM-M}}$ is stable for each $n \in \{1000,1100,1200,1300,1400\}$. This is because for the large instances, although the runtime of the partition phase is less than that of matching phase, it can not be ignored, and our algorithm only improves the matching phase of MP$_{\text{LS}}$. Thus the improvement of MP$_{\text{KM-M}}$ over MP$_{\text{LS}}$ does not increase with the increment of $n$ on large instances.

(3) The improvement of MP$_{\text{KM-M}}$ over MP$_{\text{LS}}$ is most  significant on the first set of large instances, of which the runtime is reduced by at least 19 times.
The second best improvement is on the RAND instance set, of which the runtime is reduced by at least 12 times. This is because when solving the large instances, the runtime of KM on complete bipartite graphs is longer than that on sparse bipartite graphs, and the runtime of KM on constructed bipartite graphs is longer than that on random produced bipartite graphs.

\begin{figure}[H]
    \centering
    \begin{subfigure}[b]{0.48\textwidth}
        \centering
		\setlength{\abovecaptionskip}{0.cm}
		\includegraphics[scale=0.4,trim={12 5 20 35}, clip]{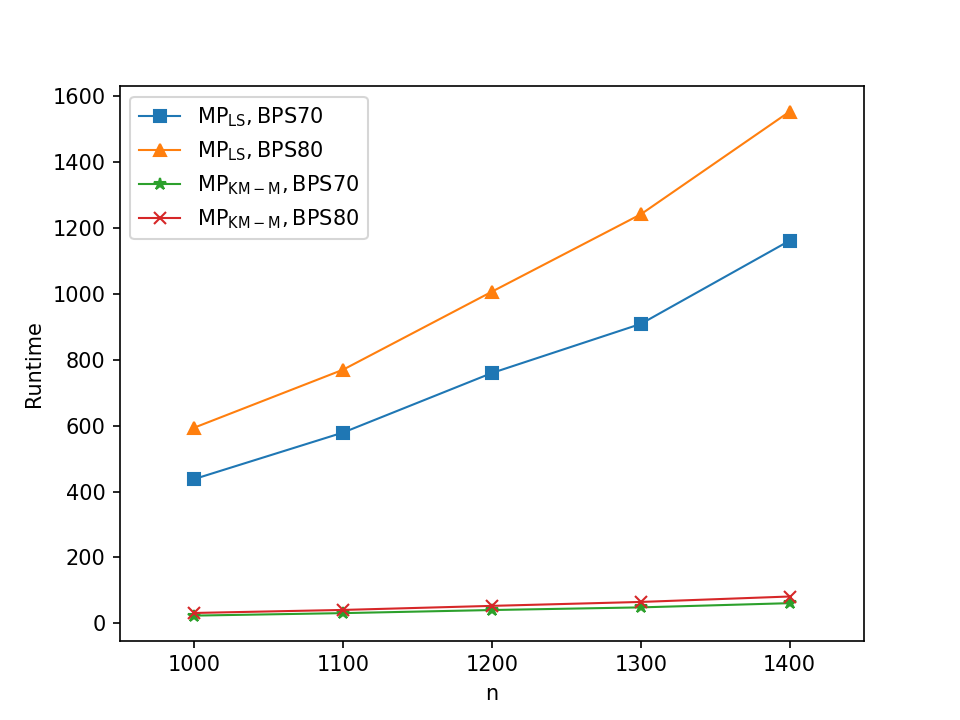}
		\caption{Comparison on the first set}
		\label{fig:lt1}
    \end{subfigure}
    \hfill
   \begin{subfigure}[b]{0.48\textwidth}
        \centering
		\setlength{\abovecaptionskip}{0.cm}
		\includegraphics[scale=0.4,trim={12 5 20 35}, clip]{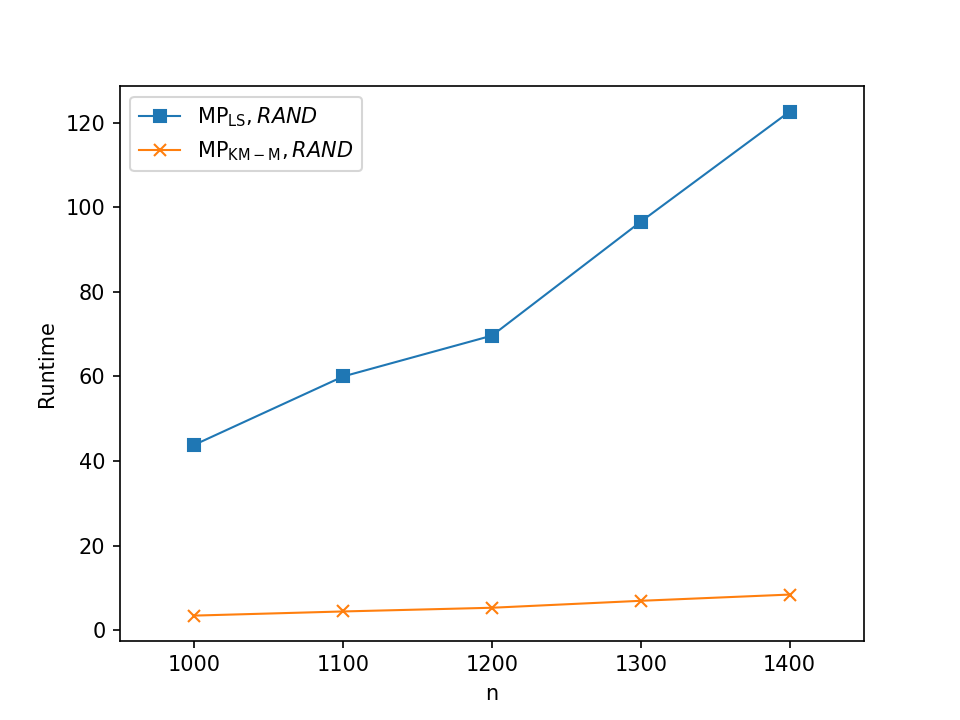}
		\caption{Comparison on the second set}
		\label{fig:lt2}
    \end{subfigure}
    \begin{subfigure}[b]{0.48\textwidth}
        \centering
		\setlength{\abovecaptionskip}{0.cm}
		\includegraphics[scale=0.4,trim={12 5 20 35}, clip]{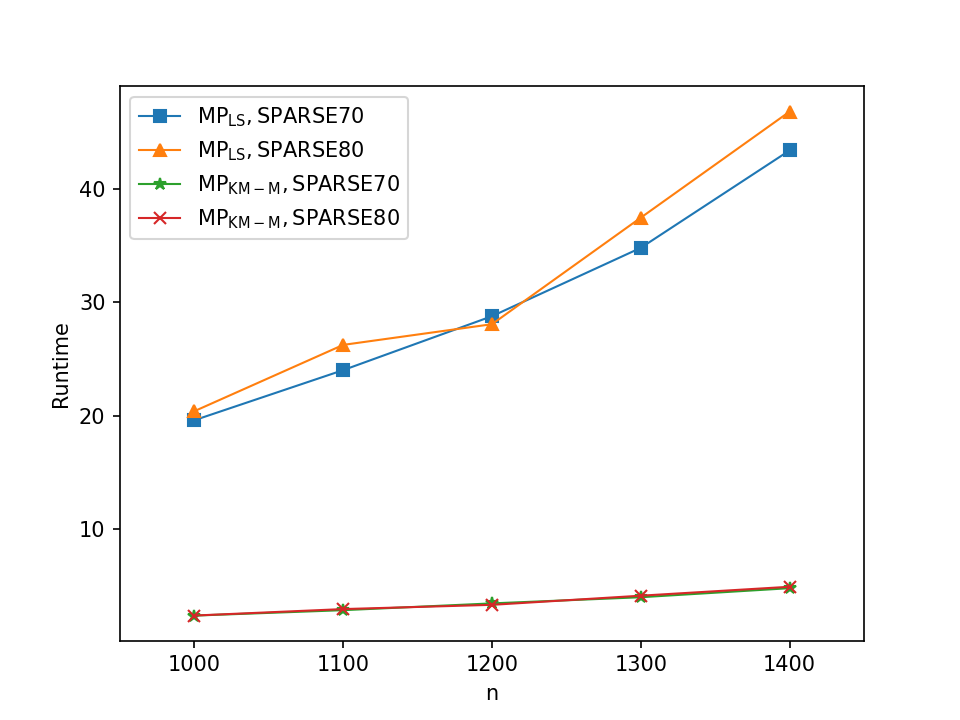}
		\caption{Comparison on the third set}
		\label{fig:lt3}
    \end{subfigure}
    \hfill
   \begin{subfigure}[b]{0.48\textwidth}
        \centering
		\setlength{\abovecaptionskip}{0.cm}
		\includegraphics[scale=0.4,trim={12 5 20 35}, clip]{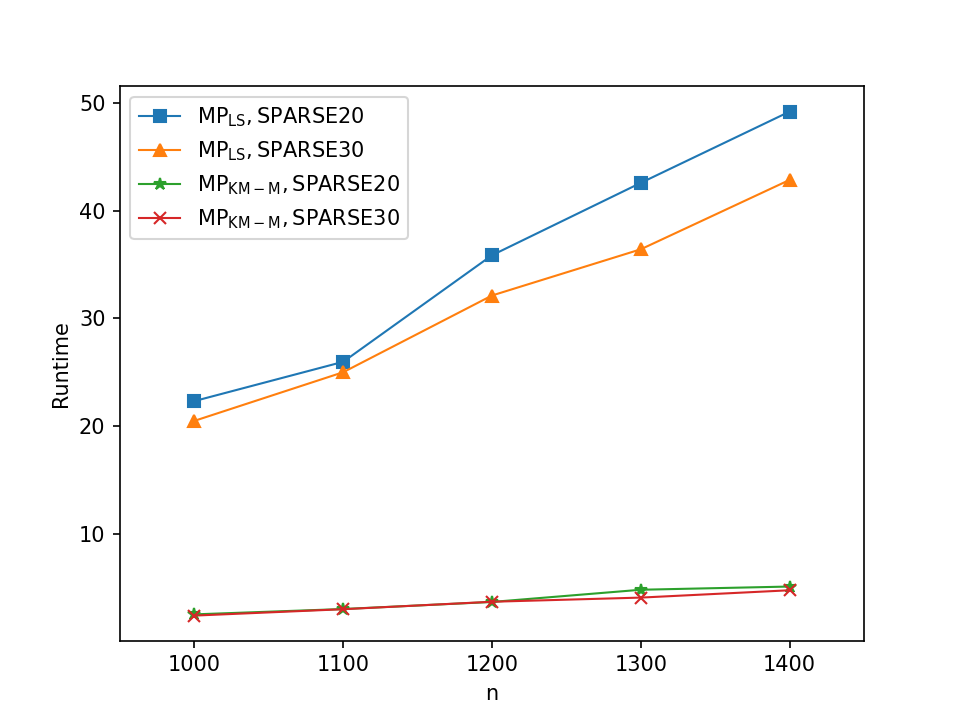}
		\caption{Comparison on the fourth set}
		\label{fig:lt4}
    \end{subfigure}
\caption{Comparison on the runtime of MP$_{\text{KM-M}}$ and MP$_{\text{LS}}$ on the large instances.}
\label{fig:lt}   
\end{figure}

\begin{figure}[ht]
	\centering
	\setlength{\abovecaptionskip}{0.cm}
	\includegraphics[scale=0.5,trim={20 10 20 40}, clip]{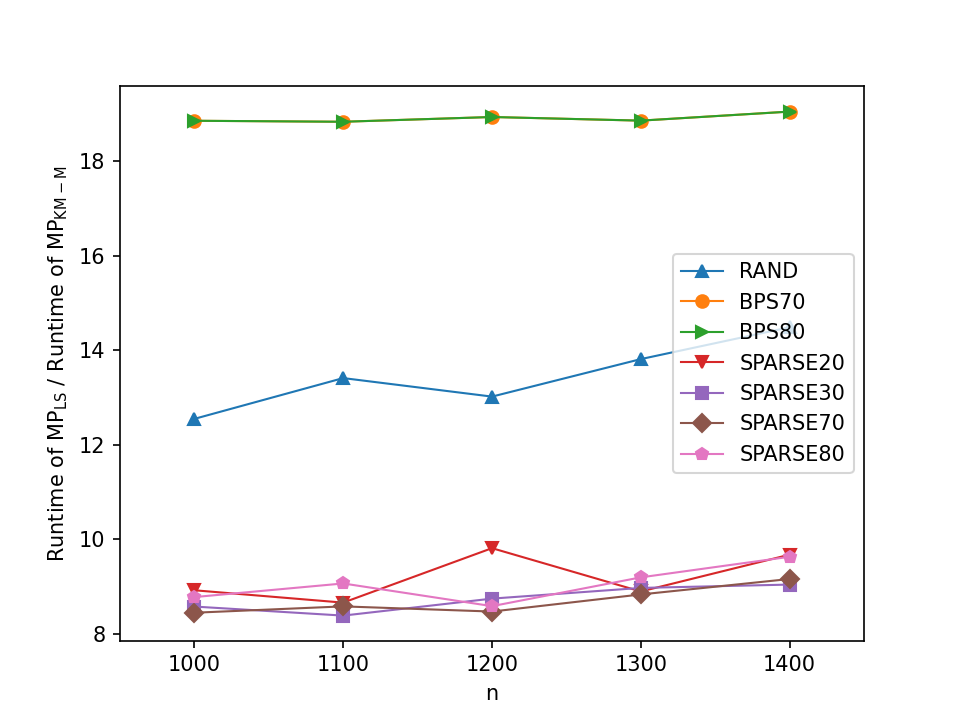}
	\caption{Ratio of the average runtime of MP$_{\text{LS}}$ over MP$_{\text{KM-M}}$ on large instances.}
	\label{fig:lt5}
\end{figure}

\section{Conclusion}
This paper proposes an iterative heuristic algorithm, denoted as MP$_{\text{KM-M}}$, for solving the Partitioning Min-Max Weighted Matching (PMMWM) problem. The proposed algorithm significantly improves the state-of-the-art PMMWM heuristic algorithm, MP$_{\text{LS}}$, on the algorithm efficiency. 
In order to escape from the local optima, at each iteration the MP$_{\text{LS}}$ algorithm adjusts the weight of an edge and generates a different initial maximum matching by the KM algorithm. However, applying KM to calculate the whole graph at each iteration is redundant and inefficiency. To address this issue, we propose a KM-M algorithm that only calculates the matching of the vertex connected with the specific edge whose weight is changed in the previous iteration, and replace the KM algorithm in MP$_{\text{LS}}$. The proposed MP$_{\text{KM-M}}$ algorithm can reduce the time complexity of the patching phase of MP$_{\text{LS}}$ from $O(n^3)$ to $O(n^2)$ without reducing the solution quality. We also proof the correctness of MP$_{\text{KM-M}}$ by the primal dual method. 
We generate extensive and diverse PMMWM instances to evaluate the performance of the proposed algorithm. The experimental results on all the 18,900 instances demonstrate the significant improvement of MP$_{\text{KM-M}}$ over MP$_{\text{LS}}$.

Our strategy is useful for other algorithms that do incremental updates during the iterations. In future work, we will try other related problems and their algorithms to try to improve the algorithm efficiency. 

\clearpage

\nocite{*}

\bibliography{00main}

\begin{thebibliography}{14}
\expandafter\ifx\csname natexlab\endcsname\relax\def\natexlab#1{#1}\fi
\providecommand{\url}[1]{\texttt{#1}}
\providecommand{\href}[2]{#2}
\providecommand{\path}[1]{#1}
\providecommand{\DOIprefix}{doi:}
\providecommand{\ArXivprefix}{arXiv:}
\providecommand{\URLprefix}{URL: }
\providecommand{\Pubmedprefix}{pmid:}
\providecommand{\doi}[1]{\href{http://dx.doi.org/#1}{\path{#1}}}
\providecommand{\Pubmed}[1]{\href{pmid:#1}{\path{#1}}}
\providecommand{\bibinfo}[2]{#2}
\ifx\xfnm\relax \def\xfnm[#1]{\unskip,\space#1}\fi
\bibitem[{Ahuja et~al.(1993)Ahuja, Magnanti \&
  Orlin}]{DBLP:books/daglib/0069809}
\bibinfo{author}{Ahuja, R.~K.}, \bibinfo{author}{Magnanti, T.~L.}, \&
  \bibinfo{author}{Orlin, J.~B.} (\bibinfo{year}{1993}).
\newblock {\it \bibinfo{title}{Network flows - theory, algorithms and
  applications}\/}.
\newblock \bibinfo{publisher}{Prentice Hall}.
\bibitem[{Barketau et~al.(2015)Barketau, Pesch \& Shafransky}]{a2}
\bibinfo{author}{Barketau, M.}, \bibinfo{author}{Pesch, E.}, \&
  \bibinfo{author}{Shafransky, Y.~M.} (\bibinfo{year}{2015}).
\newblock \bibinfo{title}{Minimizing maximum weight of subsets of a maximum
  matching in a bipartite graph}.
\newblock {\it \bibinfo{journal}{Discrete Applied Mathematics}\/},  {\it
  \bibinfo{volume}{196}\/}, \bibinfo{pages}{4--19}.
\bibitem[{Dantzig et~al.(1956)Dantzig, Ford \& Fulkerson}]{primal-dual}
\bibinfo{author}{Dantzig, G.~B.}, \bibinfo{author}{Ford, L.~R.}, \&
  \bibinfo{author}{Fulkerson, D.~R.} (\bibinfo{year}{1956}).
\newblock \bibinfo{title}{A primal-dual algorithm for linear programs}.
\newblock {\it \bibinfo{journal}{Linear Inequalities and Related Systems}\/},
  (pp. \bibinfo{pages}{171--181}).
\bibitem[{Derigs \& Zimmermann(1978)}]{DBLP:journals/computing/DerigsZ78}
\bibinfo{author}{Derigs, U.}, \& \bibinfo{author}{Zimmermann, U.}
  (\bibinfo{year}{1978}).
\newblock \bibinfo{title}{An augmenting path method for solving linear
  bottleneck assignment problems}.
\newblock {\it \bibinfo{journal}{Computing}\/},  {\it \bibinfo{volume}{19}\/},
  \bibinfo{pages}{285--295}.
\bibitem[{Edmonds \& Karp(1972)}]{DBLP:journals/jacm/EdmondsK72}
\bibinfo{author}{Edmonds, J.~R.}, \& \bibinfo{author}{Karp, R.~M.}
  (\bibinfo{year}{1972}).
\newblock \bibinfo{title}{Theoretical improvements in algorithmic efficiency
  for network flow problems}.
\newblock {\it \bibinfo{journal}{Journal of the {ACM}}\/},  {\it
  \bibinfo{volume}{19}\/}, \bibinfo{pages}{248--264}.
\bibitem[{Garey \& Johnson(1979)}]{b1}
\bibinfo{author}{Garey, M.~R.}, \& \bibinfo{author}{Johnson, D.~S.}
  (\bibinfo{year}{1979}).
\newblock {\it \bibinfo{title}{Computers and intractability: a guide to the
  theory of {NP}-Completeness}\/}.
\newblock \bibinfo{publisher}{W. H. Freeman}.
\bibitem[{Glover(1989)}]{DBLP:journals/informs/Glover89}
\bibinfo{author}{Glover, F.~W.} (\bibinfo{year}{1989}).
\newblock \bibinfo{title}{Tabu search - part {I}}.
\newblock {\it \bibinfo{journal}{{INFORMS} Journal on Computing}\/},  {\it
  \bibinfo{volume}{1}\/}, \bibinfo{pages}{190--206}.
\bibitem[{Glover(1990)}]{DBLP:journals/informs/Glover90}
\bibinfo{author}{Glover, F.~W.} (\bibinfo{year}{1990}).
\newblock \bibinfo{title}{Tabu search - part {II}}.
\newblock {\it \bibinfo{journal}{{INFORMS} Journal on Computing}\/},  {\it
  \bibinfo{volume}{2}\/}, \bibinfo{pages}{4--32}.
\bibitem[{Kress et~al.(2015)Kress, Meiswinkel \& Pesch}]{a1}
\bibinfo{author}{Kress, D.}, \bibinfo{author}{Meiswinkel, S.}, \&
  \bibinfo{author}{Pesch, E.} (\bibinfo{year}{2015}).
\newblock \bibinfo{title}{The partitioning min-max weighted matching problem}.
\newblock {\it \bibinfo{journal}{European Journal of Operational Research}\/},
  {\it \bibinfo{volume}{247}\/}, \bibinfo{pages}{745--754}.
\bibitem[{Kuhn(1955)}]{a3}
\bibinfo{author}{Kuhn, H.~W.} (\bibinfo{year}{1955}).
\newblock \bibinfo{title}{The hungarian method for the assignment problem}.
\newblock {\it \bibinfo{journal}{Naval Research Logistics Quarterly}\/},  {\it
  \bibinfo{volume}{2}\/}, \bibinfo{pages}{83--97}.
\bibitem[{Kuhn(2010)}]{DBLP:books/daglib/p/Kuhn10}
\bibinfo{author}{Kuhn, H.~W.} (\bibinfo{year}{2010}).
\newblock \bibinfo{title}{The {Hungarian} method for the assignment problem}.
\newblock In {\it \bibinfo{booktitle}{50 Years of Integer Programming 1958-2008
  - From the Early Years to the State-of-the-Art}\/} (pp.
  \bibinfo{pages}{29--47}).
\newblock \bibinfo{publisher}{Springer}.
\bibitem[{Mitchell(1998)}]{DBLP:books/daglib/0019083}
\bibinfo{author}{Mitchell, M.} (\bibinfo{year}{1998}).
\newblock {\it \bibinfo{title}{An introduction to genetic algorithms}\/}.
\newblock \bibinfo{publisher}{{MIT} Press}.
\bibitem[{Pesch et~al.(2015)Pesch, Kress \& Meiswinkel}]{7376822}
\bibinfo{author}{Pesch, E.}, \bibinfo{author}{Kress, D.}, \&
  \bibinfo{author}{Meiswinkel, S.} (\bibinfo{year}{2015}).
\newblock \bibinfo{title}{An integrated matching and partitioning problem with
  applications in intermodal transport}.
\newblock In {\it \bibinfo{booktitle}{2015 IEEE Symposium Series on
  Computational Intelligence}\/} (pp. \bibinfo{pages}{1758--1765}).
\bibitem[{Preston \& Kozan(2001)}]{PRESTON2001983}
\bibinfo{author}{Preston, P.}, \& \bibinfo{author}{Kozan, E.}
  (\bibinfo{year}{2001}).
\newblock \bibinfo{title}{An approach to determine storage locations of
  containers at seaport terminals}.
\newblock {\it \bibinfo{journal}{Computers \& Operations Research}\/},  {\it
  \bibinfo{volume}{28}\/}, \bibinfo{pages}{983--995}.

\end{thebibliography}

\end{document}